\documentclass[12pt]{article}
\usepackage[usenames]{color}
\usepackage{graphicx, subfigure}
\usepackage{amsmath, amsthm, amssymb}
\usepackage{amsfonts}
\usepackage{fullpage}
\usepackage{ifthen}
\usepackage{url}
\usepackage[sort&compress]{natbib}
\usepackage{multirow}
\usepackage{bm}
\usepackage{tikz-qtree}

\usepackage[linesnumbered,algoruled,boxed,lined,commentsnumbered]{algorithm2e}

\newtheorem{theorem}{Theorem}[section]
\makeatletter
\def\@biblabel#1{}
\makeatother

\theoremstyle{plain}
\theoremstyle{definition}
\newtheorem{example}{Example}%[section]

\theoremstyle{remark}

\newcommand{\ex}{{\sf Exp}}

\newcommand{\gam}{{\sf Gamma}}

\title{Instantaneous and limiting behavior of an $n$-node blockchain under cyber attacks from a single hacker}
\author{Xiufeng Xu$^*$\qquad Liang Hong\footnote{Department of Mathematical Sciences, The University of Texas at Dallas, 800 West Campbell Road, Richardson, TX 75080, USA.}}

\date{\today}

\usepackage{tikz}
\usetikzlibrary{automata,arrows,positioning,calc}
\usepackage{listings}
\usepackage[labelformat=empty]{caption}
\usepackage{subfig}
\usepackage{caption}
\begin{document}

\maketitle

\begin{abstract}
We investigate the instantaneous and limiting behavior of an $n$-node blockchain which is under continuous monitoring of the IT department of a company but faces non-stop cyber attacks from a single hacker. The blockchain is functional as far as no data stored on it has been changed, deleted, or locked. Once the IT department detects the attack from the hacker, it will immediately re-set the blockchain, rendering all previous efforts of the hacker in vain.  The hacker will not stop until the blockchain is dysfunctional. For arbitrary distributions of the hacking times and detecting times, we derive the limiting functional probability, instantaneous functional probability, and mean functional time of the blockchain. We also show that all these quantities are increasing functions of the number of nodes, substantiating the intuition that the more nodes a blockchain has, the harder it is for a hacker to succeed in a cyber attack.

\smallskip

\end{abstract}

\section{Introduction}
In today's society,  exchanges of information between companies and customers take various forms via the Internet. 
This mode of interaction exposes a company to potential cyber attacks.  While shutting down the computer network is not a viable option for any business,   the loss due to a cyber attack could be astronomical.  For example,  the Colonial Pipeline ransomware attack in 2021 costed the company 4.4 million dollars in ransom payment, let alone many other indirect costs.  Quite a few other costly cyber attacks are still as fresh as in recent memory.  Therefore,  a company faces the challenge of safeguarding itself from cyber attacks and maintaining its normal operations at the same time.  One plausible approach is to back up a company's data on a continuous basis. However, such a practice is prohibitively expensive to most companies.  The blockchain is one efficient alternative that can  greatly reduce the likelihood of cyber losses.  Generally speaking, a blockchain consists of several identical computers in a network; each of these computers is called a \emph{node} and stores crucial data for a company's daily operation.  Different from a traditional computer network where data is housed on a centralized server,  the blockchain stores identical data on each of its nodes, and these nodes continuously verify their data according to the majority rule. That is,  any piece of data on a node is considered valid only if it is consistent with its counterparts on at least half of the nodes of the blockchain.  Thanks to the public key encryption (e.g.  Diffie and Hellman 1976, Rivest et al. 1978, Goldwasser and Micali 1982, Goldwasser et al. 1988), it is basically useless for a hacker to steal the data stored on the blockchain.  Therefore, most hackers would either make a malicious attack such as changing or deleting the data, or make a ransom attack by locking data on the blockchain and demand a ransom from the company.  For the former type of cyber attack, which will be referred to as a \emph{destructive attack},  a hacker must hack into at least half of the nodes before he can execute his plan. For the latter type of cyber attack, which will be called a \emph{ransom attack},  a hacker must hack into all nodes and lock the entire blockchain.  Intuitively, a blockchain is much harder to hack than a traditional centralized server; see Section~3.3 for a rigorous result.  No wonder blockchain technology has been used in a wide range of applications such as cryptocurrencies, smart contracts,  and P2P insurance; for a comprehensive review of elements of blockchains and their applications, see Tama et al. (2017),  Zheng et al. (2018),  Casino et al. (2019), Kumar et al. (2020),  Yang et al. (2020), Lu et al. (2020), and Malik et al. (2021), and references therein.  

Despite the enhanced cyber security it provides, a blockchain is not immune to cyber attacks. In particular, it is still possible for a hacker to breach the security of a blockchain.  Various authors  have studied security issues of a blockchain; see, for instance, Khan and Salah (2017), Li et al. (2017), and Meng eta al. (2018).  However, no paper in the literature has investigated the operations-research-theoretic aspect of a blockchain when it is under cyber attacks.  The current paper intends to serve as the first step towards filling this lacuna in the literature.  In this article, we study an $n$-node blockchain under the cyber attack from one hacker when it is under continuous monitoring of the IT department of a company.  The blockchain is considered to be functional as long as no data stored on it has been modified, deleted, or locked.  However, the IT department will re-set the blockchain if it detects the attack from the hacker; if this occurs,  all previous efforts of the hacker will be totally forfeited and the hacker will start over again.  The hacker will not stop as far as the blockchain remains functional. Assuming the detecting time and the hacking time follow arbitrary distributions, we obtain the limiting functional probability, instantaneous functional probability, and mean functional time of the blockchain. We also show that all these quantities are a non-decreasing function of the number of nodes of the blockchain.  To our knowledge, this is the first time that rigorous results have been given to back up the intuition  that the more nodes a blockchain has, the harder it is for a hacker to break into.

The remainder of the paper is organized as follows.  In Section~2, we give a detailed description of our model and establish notational convention.  Next, in Section~3, we investigate the stochastic behavior of our proposed blockchain model. In particular, we derive formulas for the limiting functional probability,  the instantaneous functional probability, and the mean functional time of the blockchain. We also establish rigorous results to confirm the intuition that the more nodes a blockchain has, the less likely it will be hacked.  In Section~4, we provide several numerical examples based on simulation.  Finally, we conclude the paper with some remarks in Section~5. The Appendix contains some technical details.  The python code for all examples in this paper can be found at  the following website: \url{https://github.com/xuxiufeng/Blockchain_Simulations}. 

\section{Model setup and notation}

We consider a blockchain consisting of $n$ identical nodes in a given company. Without loss of generality, we assume that $n\geq 2$. From the moment it starts to operate, the blockchain is under 24-hour continuous monitoring by the IT department of the company for any potential cyber attack. There are two major types of cyber attacks for blockchains: (i) \emph{cyber destructive attacks} in which a hacker invades the blockchain and change or/and delete the data stored on the blockchain and (ii) \emph{cyber ransom attacks} where a hacker intrudes into the blockchain, locks the data stored the blockchain,  and demands ransoms from the company. In this paper, we are mainly interested in the instantaneous and limiting behavior of the $n$-node blockchain under these two types of cyber attacks from one single hacker. We assume that blockchain is under immediate cyber attacks from the hacker the moment it starts to operate. For a hacker to change or delete the data on an $n$-node blockchain, he must hack into at least $m=\lfloor n/2 \rfloor+1$ nodes. Without loss of generality, we assume that a hacker will start the malicious attack once he succeeds in hacking into $m$ nodes, and that he hacks these $m$ nodes in sequential order. That is, once the hacker penetrates the firewall of one node, he will hold onto it and immediately start hacking the next node. However, if the IT department of the company detects that the blockchain is under attack, they will immediately re-set the whole blockchain so that all previous efforts of the hacker are forfeited. Therefore, the hacker will not be able to change or delete data on the blockchain until he hacks into the $m$-th node without being detected. We assume that the hacker will perform a  cyber destructive attack the moment he succeeds in hacking into $m$ nodes without being detected. Similarly, if the hacker intends to conduct a cyber ransom attack, he will do so immediately after he hacks into the $n$-th node without being detected. On the other hand, each time the IT department re-sets the blockchain, there will be associated costs due to disruption of business operation. For this reason, the IT department will only re-set the blockchain when they detect that a cyber attack against the blockchain is under way.

Without loss of generality, we label the $n$ nodes of the blockchain as node $1$, nodes $2$, $\ldots$, node $n$. We assume that it takes a random amount of time $X_i$ for the hacker to hack into node $i, (i=1, \ldots, n)$, and that $X_1, \ldots, X_n$ are independent and identically distributed (iid) according to a distribution function $F_X$ with a probability density function $f_X(x)$. We also assume that it takes a random amount of time $Y$ for the IT department to detect that the blockchain is under a cyber attack from the moment it starts to operate or has just been  re-set, that $Y$ has a distribution function $F_Y$ with a probability density function $f_Y(y)$, and that $Y$ is independent from $X_1, \ldots, X_n$. We call each $X_i$ a \emph{hacking time} and $Y$ a \emph{detecting time}.

The stochastic behavior of the blockchain under a cyber destructive attack can be described as follows. At time $t=0$, the blockchain is put into operation and immediately under continuous monitoring for cyber attacks. At the same time, a hacker starts to work on bypassing the firewall of a node. Since all $n$ nodes are identical, we may assume that the hacker attacks the nodes of the blockchain according to the ascending order of their labels.  More specifically, the hacker attacks node $1$ first. If $X_1<Y$, i.e., the hacker penetrates the firewall of node $1$ without being detected, then he will immediately attack node $2$. On the other hand, if $X_1 \geq Y$, i.e., the hacker fails to bypass the firewall of node $1$ before being detected, the blockchain will be re-set and he has to start over. In general, if $\sum_{i=1}^m X_i <Y$, i.e., the hacker is able to hack into $m$ nodes before being detected, he will perform a cyber destructive attack. But, if $\sum_{i=1}^m X_i\geq Y$, i.e., he is being detected before he is able to change any data on the blockchain, the blockchain will be re-set, rendering  all previous efforts of the hacker in vain. As far as the hacker has not been able to change or delete any data on the blockchain, we say the blockchain is \emph{functional}. The stochastic behavior of the blockchain under a cyber ransom attack can be described in a similar manner by replacing $m$ with $n$. Therefore, we will focus only on the case of a cyber destructive attack.  However, all results extend to the case of a cyber ransom attack if we replace $m$ with $n$.  

We are mainly interested in the limiting behavior and the instantaneous behavior of an $n$-node blockchain in the aforementioned setup. Recall that $m=\lfloor n/2\rfloor+1$.  For a given $m$, we put

\begin{equation}
\label{eq:limit}
P_m(\infty)=P\{\text{the  blockchain will never be hacked}\}.
\end{equation} 
We say $P_m(\infty)$ is the \emph{limiting functional probability} of an $n$-node blockchain. We also define
$T_m$ to be the total amount of time it takes for the hacker to hack into the blockchain, and call $E[T_m]$ the \emph{mean functional time} of the blockchain.  Moreover, for $t>0$, we define

\begin{equation}
\label{eq:instantaneous}
P_m(t)=P\{\text{the blockchain has not been hacked at time $t$}\}.
\end{equation}
We call $P_m(t)$ the \emph{instantaneous functional probability} of the blockchain. It follows from the continuity of the probability measure (e.g.  Kallenberg 2002, Lemma~1.14) that the limiting functional probability and the instantaneous functional probability can be related as follows

\[P_m(\infty)=\lim_{t\rightarrow \infty} P_m(t).\]

\section{Stochastic behavior of the proposed blockchain model}
\subsection{Limiting functional probability and mean functional time}

To investigate the limiting behavior of the blockchain, we define a \emph{cycle} to be the period from the moment the blockchain starts afresh or has just been re-set to the next moment it is being re-set. That is, a cycle is a period from the inception of a detecting time to its end, provided the blockchain has not been hacked. Let $N_1$ be the number of cycles needed for the hacker to succeed in hacking into the blockchain. Then $N_1$ is a geometric random variable\footnote{Here we interpret a geometric random variable with parameter $0<p<1$ to be the number of failures (rather than the number of trials) until the first success among a sequence of independent Bernoulli trials with success probability $p$.} with parameter $p_m$ given by
\begin{equation}
\label{eq:detectprob}
p_m=P\left\{ Y>\sum_{i=1}^m X_i \right\}=\int_0^{\infty} F_X^{(m)}(s) dF_Y(s) ,
\end{equation}
where $F_X^{(m)}$ is the $m$-folded convolution of $F_X$. Note that $p_m>0$ since we assume both $F_X$ and $F_Y$ admit a probability density function. If we restrict our attention to the epochs the blockchain starts afresh, or has just been re-set, or has just been hacked, then we obtain an embedded discrete-time process. This discrete-time process is a two-state discrete-time Markov chain. In State 1, the blockchain is functional, but in State 2, the blockchain is dysfunctional, i.e., it has been hacked. A diagram of this two-state discrete-time Markov chain is given in Figure 1.\\

\begin{center}
	\begin{tikzpicture}[->, >=stealth', auto, semithick, node distance=4cm]
	\tikzstyle{every state}=[fill=white,draw=black,thick,text=black,scale=2]
	\node[state]    (A)                     {$1$};     
	\node[state]    (B)[right of=A]   {$2$};
	\path
	(A) edge[loop above]			node{$1-p$}	(A)
	       edge[bend left,above]	node{$p$}	(B)
	(B) edge[loop above]          node{$1$}(B);   
	\end{tikzpicture}          
\end{center}
\begin{center}
	Figure 1: Two-state discrete Markov chain
\end{center}
Note that State 2 is an absorbing state. It is clear that the transition probability from State 1 to State 2 is exactly the probability $p_m$ given by (\ref{eq:detectprob}). Conditioning on the outcome of the first cycle, we have
\[P_m(\infty)=p_m\times 0+(1-p_m)\times P_m(\infty), \]
implying that $P_m(\infty)=0$, which is in agreement with our intuition: if we observe the stochastic behavior of the blockchain over an infinite horizon of time, then it will eventually be hacked because $p_m>0$.

Next, we look at the mean functional time of the blockchain. To this end, we focus only on the moments every possible cycle starts or ends.
By doing so, we identify a renewal process $\{N_2(t)\}_{t\geq 0}$ where
\[N_2(t)= \max\left\{k: \sum_{i=1}^k Y_i\leq t \right\},\]
and $Y_i\overset{d}{=}Y\mid Y<\sum_{i=1}^m X_i$.
\begin{equation}
\label{eq:totaltime}
T_m=\sum_{i=1}^{N_1} Y_i +\left(\sum_{i=1}^m X_i  \middle|\ \sum_{i=1}^m X_i<Y\right).
\end{equation}
Also, it is clear that $N_1$ is independent of $Y_{n+1}, Y_{n+2}, \ldots$, which shows $N_1$ is a stopping time for the sequence
$\{Y_i\}_{i\geq 1}$. Therefore, Wald's Identity implies that
\begin{eqnarray}
\label{eq:meantime}
E[T_m] &=& E[N_1] E[Y_1]+E \left[ \sum_{i=1}^m X_i \middle|\ \sum_{i=1}^m X_i<Y\right ] \nonumber \\
     &=& E[N_1] E\left [Y\middle|\ Y<\sum_{i=1}^m X_i \right]+E \left[ \sum_{i=1}^m X_i \middle|\ \sum_{i=1}^m X_i<Y\right ].
\end{eqnarray}

\subsection{Instantaneous functional probability}

To obtain the instantaneous functional probability of the blockchain,  we let $S_t$ be the moment the last time the blockchain is re-set before time $t$:
\[S_t=\sum_{i=1}^{N_2(t)} Y_i,\]
and $F_{S_t}$ be its distribution function. Then for $t>0$, we have
\begin{eqnarray*}
P_m(t) &=& P\{\text{the blockchain is functional at $t$}\mid S_t=0\}P\{S_t=0\}\\
     & & +\int_0^t P\{ \text{the blockchain is functional at $t$}\mid S_t=s\}dF_{S_t}(s)\\
     &=& P\{t<\sum_{i=1}^m X_i \wedge Y \}+\int_0^t P\{ \text{the blockchain is functional at $t$}\mid S_t=s\} dF_{S_t}(s).
\end{eqnarray*}
For $0<s<t$, we have
\[P\{ \text{the blockchain is functional at $t$}\mid S_t=s\}=P\{Y \wedge \sum_{i=1}^m X_i>t-s \mid Y_1>t-s\}, \]
and $dF_{S_t}(s)=P\{Y_1>t-s\}dG(s)$ where $G(s)=\sum_{i=1}^{\infty} F_{Y_i}(s)$.\\

Therefore,
\begin{equation}
\label{eq:instantaneous2}
P_m(t) = P\{t<\sum_{i=1}^m X_i \wedge Y \}+\int_0^t P\{Y \wedge \sum_{i=1}^m X_i>t-s \} dG(s).
\end{equation}
It is clear that the first term in (\ref{eq:instantaneous2}) goes to $0$ as $t\rightarrow \infty$.
For the second term, we have
\begin{eqnarray*}
0 &\leq & \lim_{t \rightarrow \infty} \int_0^t P\{Y \wedge \sum_{i=1}^m X_i>t-s \} dG(s)\\
   &=&     \lim_{t \rightarrow \infty} \int_0^{\infty} P\{Y \wedge \sum_{i=1}^m X_i>t-s \}1_{[0, t]}(s) dG(s)\\
   &\leq & \lim_{t \rightarrow \infty} \int_0^{\infty} P\{Y \wedge \sum_{i=1}^m X_i>t-s\} dG(s)\\
   &= &  \int_0^{\infty} \lim_{t \rightarrow \infty} P\{Y \wedge \sum_{i=1}^m X_i>t-s \} dG(s)=0, \\
\end{eqnarray*}
where the last equality follows from the monotone convergence theorem for sequences of decreasing measurable functions. This gives another proof that $P_m(\infty)=0$.

\subsection{$P_m(t)$ and $E[T_m]$ as functions of $m$}

Intuitively, the more nodes a blockchain has, the more secure it is.  That is, $P_n(t)$,  $P_n(\infty)$, and $E[T_n]$ should all  be non-decreasing functions of the number of nodes $n$, and they should be increasing functions of $n$ if the increment of the number of nodes is at least $2$. The next theorem confirms this intuition.

\begin{theorem}
\label{theorem3.1}
$P_m(t)$,  $P_m(\infty)$, and $E[T_m]$ are all  increasing functions of  $m$, where $m=\lfloor n/2\rfloor+1$. 
\end{theorem}

\begin{proof}
First, it is obvious that 
\begin{equation*}
P\{t<Y \wedge \sum_{i=1}^{m+1}X_i \}> P\{t<Y \wedge \sum_{i=1}^{m}X_i\},
\end{equation*}
and
\begin{equation*}
P\{t-s<Y \wedge \sum_{i=1}^{m+1}X_i \}> P\{t-s<Y \wedge \sum_{i=1}^{m}X_i\}
\end{equation*}
because each hacking time $X_i$ is positive almost surely.  Therefore, (\ref{eq:instantaneous2}) implies
 \begin{equation}
\label{eq:instantaneous3}
 	P_{m+1}(t)> P_m(t), \quad \text{for all positive integer $m$}.
 \end{equation}

Next, we write (\ref{eq:meantime}) as 
\begin{eqnarray*}
E[T_m] &=& \frac{1-p_m}{p_m}\times \frac{E[Y1_{\{Y<\sum_{i=1}^m X_i \}}]}{1-p_m}+\frac{E[(\sum_{i=1}^m X_i)1_{\{Y>\sum_{i=1}^m X_i\}}]}{p_m}\\
&=& \frac{E[Y1_{\{Y<\sum_{i=1}^m X_i\}}]}{p_m}+\frac{E[(\sum_{i=1}^m X_i)1_{\{Y>\sum_{i=1}^m X_i\}}]}{p_m}.
\end{eqnarray*}
Therefore, 
\begin{eqnarray*}
 E[T_{m+1}]-E[T_m] 
&=&  (p_mp_{m+1})^{-1} \bigg \{ p_m E[Y1_{\{Y<\sum_{i=1}^{m+1} X_i \}}]-p_{m+1}E[Y1_{\{Y<\sum_{i=1}^{m} X_i \}}] \\
& &+p_m E\bigg[\bigg(\sum_{i=1}^{m+1} X_i\bigg)1_{\{Y>\sum_{i=1}^{m+1} X_i\}}\bigg]-p_{m+1}E\bigg[\bigg(\sum_{i=1}^m X_i\bigg)1_{\{Y>\sum_{i=1}^m X_i\}}\bigg] \bigg\}.
\end{eqnarray*}
The four terms in the brackets can be written as 
\begin{eqnarray*}
& & p_m E[Y1_{\{Y<\sum_{i=1}^m X_i \}}]-p_{m+1}E[Y1_{\{Y<\sum_{i=1}^{m} X_i \}}]+p_m E[Y1_{\{\sum_{i=1}^m X_i<Y<\sum_{i=1}^{m+1} X_i \}}] \\
& & + p_m E[X_{m+1}1_{\{Y>\sum_{i=1}^{m+1} X_i\}}]+ p_m E\bigg[\bigg(\sum_{i=1}^m X_i\bigg)1_{\{Y>\sum_{i=1}^{m+1} X_i\}}\bigg]\\
& &-p_{m+1} E\bigg[\bigg(\sum_{i=1}^m X_i\bigg)1_{\{Y>\sum_{i=1}^{m+1} X_i\}}\bigg] 
 -p_{m+1} E\bigg[\bigg(\sum_{i=1}^m X_i\bigg)1_{\{\sum_{i=1}^{m+1}X_i>Y>\sum_{i=1}^m X_i\}}\bigg]\\
&=& (p_m-p_{m+1})E[Y1_{\{Y<\sum_{i=1}^m X_i \}}]+(p_m-p_{m+1})E\bigg[\bigg(\sum_{i=1}^m X_i\bigg)1_{\{Y>\sum_{i=1}^{m+1} X_i\}}\bigg]\\
& &  +p_m E[Y1_{\{\sum_{i=1}^{m+1}X_i>Y>\sum_{i=1}^m X_i\}}]\ -p_{m+1} 
E\bigg[\bigg(\sum_{i=1}^m X_i\bigg)1_{\{\sum_{i=1}^{m+1}X_i>Y>\sum_{i=1}^m X_i\}}\bigg] \\
& & +p_m E[X_{m+1}1_{\{Y>\sum_{i=1}^{m+1} X_i\}}] \\
&\geq& (p_m-p_{m+1})E[Y1_{\{Y<\sum_{i=1}^m X_i \}}]+(p_m-p_{m+1})E\bigg[\bigg(\sum_{i=1}^mX_i\bigg)1_{\{Y>\sum_{i=1}^{m+1} X_i\}}\bigg]\\
& &   +p_m E\bigg[\bigg(\sum_{i=1}^m X_i\bigg)1_{\{\sum_{i=1}^{m+1}X_i>Y>\sum_{i=1}^m X_i\}}\bigg]\ -p_{m+1} 
E\bigg[\bigg(\sum_{i=1}^m X_i\bigg)1_{\{\sum_{i=1}^{m+1}X_i>Y>\sum_{i=1}^m X_i\}}\bigg]\\
& & +p_m E[X_{m+1}1_{\{Y>\sum_{i=1}^{m+1} X_i\}}]\\
&=& (p_m-p_{m+1})E[Y1_{\{Y<\sum_{i=1}^m X_i \}}]+(p_m-p_{m+1})E\bigg[\bigg(\sum_{i=1}^m X_i\bigg)1_{\{Y>\sum_{i=1}^{m+1} X_i\}}\bigg]\\
& &   +(p_m-p_{m+1}) E\bigg[\bigg(\sum_{i=1}^m X_i\bigg)1_{\{\sum_{i=1}^{m+1}X_i>Y>\sum_{i=1}^m X_i\}}\bigg] +p_m E[X_{m+1}1_{\{Y>\sum_{i=1}^{m+1} X_i\}}].
\end{eqnarray*}
It is evident from (\ref{eq:detectprob}) that $p_m$ is decreasing in $m$.  Therefore,  $ E[T_{m+1}]-E[T_m] >0$. 
\end{proof}

In view of the above results, one would expect it to be impossible for a blockchain to be hacked when the number of nodes $n\rightarrow \infty$.  The next theorem confirms this intuition.  (Note that $n\rightarrow \infty$ if and only if $m=\lfloor n/2\rfloor+1\rightarrow \infty$.)

\begin{theorem}
\label{theorem3.2}
$\lim_{m\rightarrow\infty}P_m(t)=1$ and $\lim_{m\rightarrow\infty}E[T_m]=\infty$, where $m=\lfloor n/2\rfloor+1$. 
\end{theorem}

\begin{proof}
Since each $X_i$ is positive almost surely,  (\ref{eq:instantaneous2}) implies
\begin{equation}
\label{eq:limitofprob}
	\lim_{m \rightarrow \infty}P_m(t)=1. 
\end{equation}
Next, (\ref{eq:detectprob}) implies $p_m\rightarrow 0$ as $m\rightarrow\infty$. Therefore,  $E[N_1]=1/p_m-1\rightarrow \infty$ as $m\rightarrow\infty$.  It follows from (\ref{eq:meantime}) that 
\begin{eqnarray*}
\lim_{m \rightarrow \infty}E[T_m ]
%&=& \lim_{n \rightarrow \infty}\left( E[N_1] E\left [Y\middle|\ Y<\sum_{i=1}^m X_i \right]+E \left[ \sum_{i=1}^m X_i \middle|\ \sum_{i=1}^m X_i<Y\right ] \right) \\
&\geq&	\lim_{m \rightarrow \infty}E[N_1] E\left [Y\middle|\ Y<\sum_{i=1}^m X_i \right]\\
&=& \lim_{m \rightarrow \infty}E[N_1] \lim_{m \rightarrow \infty}\left[ \frac{E[Y1_{\{Y<\sum_{i=1}^m X_i\}} ]}{P\left\{ Y<\sum_{i=1}^m X_i \right\}}\right] \\
&=&\lim_{m \rightarrow \infty}E[N_1] E[Y]\\
&=& \infty.
\end{eqnarray*}
\end{proof}

\subsection{Cost-benefit analysis}

Suppose $C_1(m)$ is the cost incurred each time the blockchain is re-set, $C_2(m)$ is the cost per unit time for running  the blockchain,  and $R$ is the revenue earned per unit time by running the blockchain.  Note that $C_1(m)$ is incurred exactly once in each cycle because the blockchain is re-set exactly once in each cycle.  By Wald's Identity, the expected number of times $C_1(m)$ is incurred in one unit time equals $1/E[Y]$. Therefore,  the expected net revenue (i.e.  expected profit) per unit time, denoted as $E_m[NR]$, can be written as
\begin{equation}
\label{eq:revenue}
E_m[NR]=R-C_2(m)-\frac{C_1(m)}{E[Y]}.
\end{equation}
In practice, $C_2(m)$ is an increasing function of $m$, but $1/E[Y]$  is a decreasing function of $m$.  While an optimal $m$ that maximizes the net revenue per unit time does not necessarily exist,  one can always use (\ref{eq:revenue})   to perform a cost-benefit analysis.

\section{Examples}
\begin{example}[Exponential hacking- and exponential detecting times]
\label{example1}

Suppose $X_1,\ldots, X_n\overset{iid}{\sim} \ex(\lambda)$, $Y\sim \ex(\delta)$,  and $Y$ is independent from $X_1,  \ldots, X_n$, where $\ex(\lambda)$ and $\ex(\delta)$ denote the exponential distributions with rate $\lambda>0$ and $\delta>0$, respectively. That is,  $X\sim \ex(\lambda)$ if and only if its  probability density function is given by
\[f(x)=\lambda e^{-\lambda x}, \quad x>0.\] 
It is clear that  $\sum_{i=1}^m X_i \sim \gam(m, \lambda)$, where $\gam(m, \lambda)$ denotes the gamma distribution with shape parameter $m$ and rate parameter $\lambda$. That is, $X\sim \gam(\alpha, \lambda)$ if and only if its  probability density function is given by 
\[f(x)=\frac{\lambda^\alpha}{\Gamma(\alpha)}x^{\alpha-1}e^{-\lambda x}, \quad x>0,\]
where $\Gamma(\alpha)=\int_0^\infty t^{\alpha-1}e^{-t} dt$ is the Gamma function. 
It follows that
\begin{align*}
	E[T_m]=&\frac{\lambda^m \bigg(\int_0^\infty \int_y^\infty ye^{-(\delta y+\lambda s)}s^{m-1}dsdy+\int_0^\infty \int_s^\infty e^{-(\delta y+\lambda s)}s^m dyds\bigg)}{\int_0^\infty \gamma(m,\lambda s)e^{-\delta s}ds},\\
	P_m(t)=&e^{-\delta t}-\frac{1}{\Gamma(m)}\gamma(m, \lambda t)+\frac{1}{\Gamma(m)}\gamma(m, \lambda t)(1-e^{-\delta t}),\\
	&+\int_0^t e^{-\delta (t-s)}-\frac{1}{\Gamma(m)}\gamma(m, \lambda (t-s))+\frac{1}{\Gamma(m)}\gamma(m, \lambda (t-s))(1-e^{-\delta (t-s)})\\
	&d\sum_{i=1}^\infty \bigg\{\frac{1-e^{-\delta s}}{1-\int_0^\infty \frac{1}{\Gamma(m)}\gamma(m, \lambda w)\delta e^{-\delta w}dw}1_{0\leq s \leq\sum_{i=1}^{m}X_i}+1_{s >\sum_{i=1}^{m}X_i}\bigg\},
\end{align*}
where $\gamma(\alpha, \lambda) =\int_0^{\lambda}t^{\alpha-1}e^{-t}dt $ is the lower incomplete gamma distribution.  The Appendix contains the derivation of the above formulas.

%%%%%
\begin{figure}%[h]
\begin{center}
\subfigure[$E(T_m)$ as a function of $m$]{\centering \includegraphics[scale=0.45]{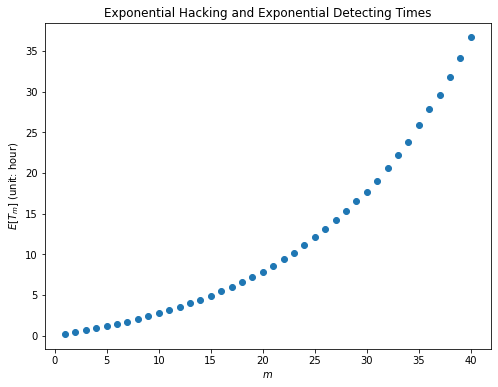}}
\subfigure[$P_m(5)$ as a function of $m$]{\centering \includegraphics[scale=0.45]{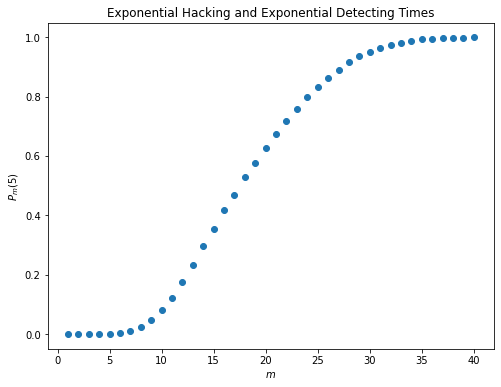}}
\end{center}
\caption{Figure 1: Plots of $E[T_m]$ and $P_m(5)$  under Exponential hacking- and exponential detecting times}
\label{fig:exp}
\end{figure}
%%%%%%

Next, we perform a simulation study to estimate $E[T_m]$ and $P_m(t)$ as functions of $m$ and $t$ respectively.  For $\lambda=0.2$ and $\delta=3$, we perform $N_1=30, 000$ iterations to estimate each $E[T_m]$ for $m=1,2,...,40$. For each $i=1, 2, \ldots N_1$, we generate an independent random amount of time $T_m^i$ for a successful cyber destructive attack.  Then we estimate $E[T_m]$ as $\sum_{i=1}^{N_1} T_m^i/N_1$. Similarly, we perform $N_2=50,000$ iterations to estimate each $P_m(5)$ for $m=1,2,...,40$.  For each $j=1, 2, \ldots, N_2$, we set $W_j=1$ if the blockchain is functional at $t=5$ and $W_j=0$ otherwise. Then we take $\sum_{j=1}^{N_2}W_j/N_2$ to be the estimate of $P_m(5)$.  Figure~\ref{fig:exp} displays the values of $E[T_m]$ and $P_m(5)$ under Exponential hacking- and exponential detecting times for various values of $m$.  In both panels, the curve is non-decreasing in $m$.  Also,  $P_m(5)$ approaches $1$ as $m$ increases. These observations are all in agreement with Theorem~\ref{theorem3.1} and Theorem~\ref{theorem3.2}.
 
%%%%%%
\begin{figure}%[h]
\begin{center}
	 \subfigure[$P_m(t)$ as a bivariate function of $m$ and $t$]{\centering \includegraphics[width=0.53\textwidth]{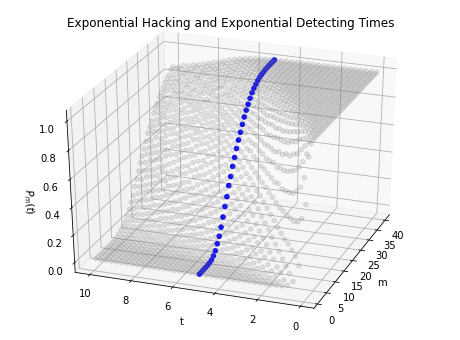}}
	 \subfigure[$E_m(NR)$ as a function of $m$]{\centering \includegraphics[width=0.46\textwidth]{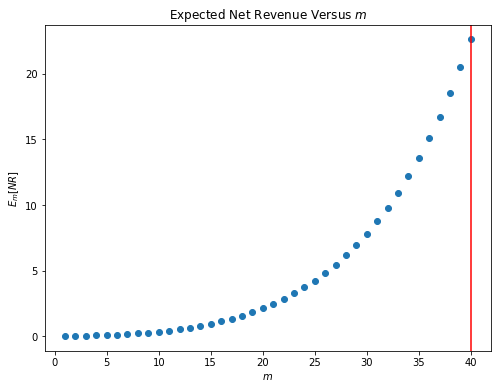}}
\end{center}
\caption{Figure~2: Plots of $P_m(t)$ and  $E_m(NR)$ under Exponential hacking- and exponential detecting times}
\label{fig:exp3D}
\end{figure}
%%%%%%
Figure~\ref{fig:exp} only demonstrates the relationship between $P_m(t)$ and $m$ when $t=5$. What can we say when $t$ varies?  In Section~3.2, we established that $P_m(\infty)=0$. Intuitively,  as $t$ increases, $P_m(t)$ would decrease for any given $m$.  Similar to the above,  we perform a simulation with $N_3=10,000$ iterations to estimate each $P_m(t)$ for various values of $m$ and $t$.  Panel~(a) of Figure~\ref{fig:exp3D} provides a 3D scatter plot of $P_m(t)$ as a bivariate function of $m$ and $t$.  The blue scatter plot in Panel~(a) of Figure~\ref{fig:exp3D} is the 3D version of $P_m(t)$ versus $m$ when $t=5$,  i.e., the plot in Panel(b) of Figure~\ref{fig:exp} corresponds to the cross section of the 3D scatter plot in Panel~(a) of Figure~\ref{fig:exp3D}.  It is clear  that $P_m(t)$ is non-decreasing in $m$  for a given $t>0$ and non-increasing in $t$ for a given $m$.  That is, the larger $m$ and $t$ are, the less likely the blockchain will be hacked. 

For cost-benefit analysis, we let $R=10^{-2}e^{\sqrt{1.5m}}$, $C_1=10^{-3}m$, and $C_2=10^{-2}m$. Then the expected net revenue equals
\begin{align*}
	E_m[NR]&=10^{-2}e^{\sqrt{1.5m}}-10^{-2}m-\frac{10^{-3}m}{E[Y]}.
\end{align*}
A plot of the expected net revenue $E_m(NR)$ as a function of $m$ is given in Panel~(b) of Figure~\ref{fig:exp3D}.  In this case,  $E_m(NR)$ is an increasing function of $m$. Therefore,  one may not find the optimal $m$  that maximizes the expected net revenue.

\end{example}

\begin{example}[Gamma hacking- and gamma detecting times]
\label{example2}

Suppose $Y\sim \gam(\alpha, \beta)$, $X_1, \ldots, X_n \overset{iid}{\sim} \gam(\eta, \delta)$, and $Y$ is independent from $X_1, \ldots, X_n$.  It follows that $\sum_{i=1}^m X_i \sim \gam(m\eta, \delta)$.  Then
\begin{align*}
E[T_m]&=\frac{\delta^{m\eta}\bigg(\int_0^\infty \int_y^\infty y^{\alpha} e^{-(\beta y + \delta s)}s^{m\eta-1}dsdy+\int_0^\infty \int_s^\infty y^{\alpha-1}e^{-(\beta y+\delta s)}s^{m\eta} dyds\bigg )}{\int_0^\infty \gamma(m\eta, \eta s)s^{\alpha -1}e^{-\beta s}ds},
\end{align*}
and
\begin{align*}
	P_m(t)=&1-\frac{1}{\Gamma(m\eta)}\gamma(m\eta, \delta t)-\frac{1}{\Gamma(\alpha)}\gamma(\alpha, \beta t)+\frac{1}{\Gamma(m\eta)}\gamma(m\eta, \delta t)\frac{1}{\Gamma(\alpha)}\gamma(\alpha, \beta t)\\
	&+\int_0^t \bigg[ 1-\frac{1}{\Gamma(m\eta)}\gamma(m\eta, \delta (t-s))-\frac{1}{\Gamma(\alpha)}\gamma(\alpha, \beta (t-s))\\
	&+\frac{1}{\Gamma(m\eta)}\gamma(m\eta, \delta (t-s))\frac{1}{\Gamma(\alpha)}\gamma(\alpha, \beta (t-s)) \bigg]\\
	&d\sum_{i=1}^\infty \bigg\{\frac{\frac{1}{\Gamma(\alpha)}\gamma(\alpha, \beta s)}{1-\int_0^\infty \frac{1}{\Gamma(m\eta)}\gamma(m\eta,\delta w)\frac{\beta^\alpha}{\Gamma(\alpha)}w^{\alpha-1}e^{-\beta w}dw}1_{\{0\leq s \leq\sum_{i=1}^{m}X_i\}}+1_{\{s >\sum_{i=1}^{m}X_i\}}\bigg\}.
\end{align*}
%where $\gamma(\alpha, \lambda)$ is the same incomplete lower gamma distribution as in Example 1. 
The derivation of the above formulas is delegated to the Appendix. 

%%%%%
\begin{figure}%[h]
\begin{center}
\subfigure[$E(T_m)$ as a function of $m$]{\centering \includegraphics[scale=0.45]{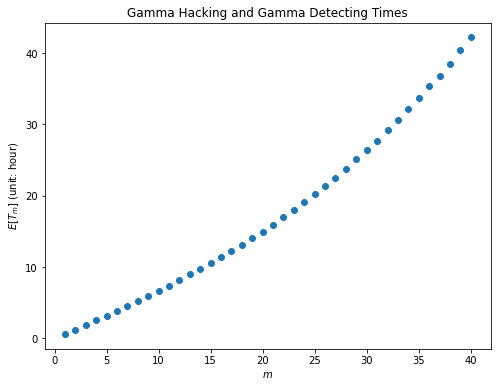}}
\subfigure[$P_m(4)$ as a function of $m$]{\centering \includegraphics[scale=0.45]{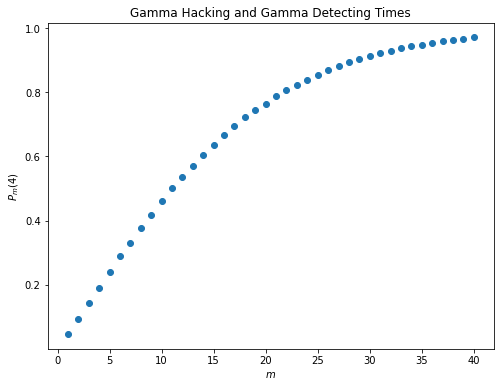}}
\end{center}
\caption{Figure 3: Plots of $E[T_m]$ and $P_m(4)$  under gamma hacking- and gamma detecting times}
\label{fig:gamma}
\end{figure}
%%%%%
As in Example~\ref{example1}, we perform a simulation study to estimate $E[T_m]$ and $P_m(t)$. For $\eta=0.05, \delta=15, \alpha=2,$ and $\beta=10$, we perform $N_1=80,000$ iterations to estimate each $E[T_m]$ for $m=1,2,...,40$. We generate an independent random amount of time $T_m^i$ for $i=1,2,\ldots,N_1$, and estimate $E[T_m]$ as $\sum_{i=1}^{N_1}T_m^i/N_1$. Similarly, with $N_2=150,000$ iterations, we estimate each $P_m(4)$ for $m=1,2,...,40$. For each $j=1,2,...,N_2$, we set $W_j=1$ if the blockchain is function at $t=4$ and $W_j=0$ otherwise. Then $P_m(4)$ is estimated by $\sum_{j=1}^{N_2} W_j/N_2$. Figure~\ref{fig:gamma} provides plots of $E[T_m]$ and $P_m(4)$ for various values of $m$.  These plots show that $E[T_m]$ and $P_m(4)$ are both increasing in $m$, and $P_m(4)$ approaches $1$ as $m$ increases.  

%%%%%%
\begin{figure}%[h]
\begin{center}
	 \subfigure[$P_m(t)$ as a bivariate function of $m$ and $t$]{\centering \includegraphics[width=0.53\textwidth]{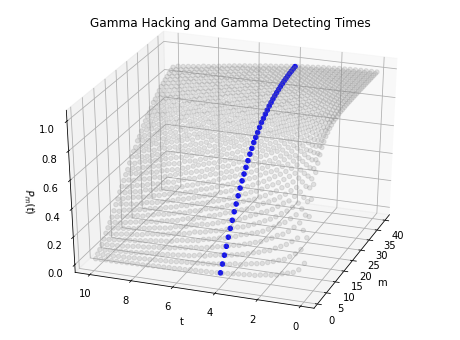}}
	 \subfigure[$E_m(NR)$ as a function of $m$]{\centering \includegraphics[width=0.46\textwidth]{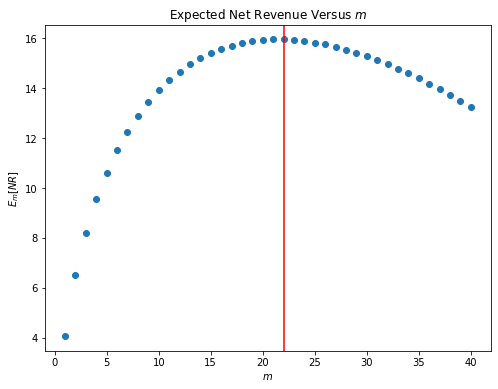}}
\end{center}
\caption{Figure 4: Plots of $P_m(t)$ and  $E_m(NR)$  under Gamma hacking- and gamma detecting times}
\label{fig:gamma3D}
\end{figure}
%%%%%%

Panel~(a) of Figure~\ref{fig:gamma3D} provides a 3D plot of $P_m(t)$ as a bivariate function of $m$ and $t$ under gamma hacking- and gamma detecting times.  Similar to the previous example, the blue scatter plot in Figure~\ref{fig:gamma3D} is the 3D version (cross section) of $P_m(t)$ versus $m$ when $t=4$. It is clear that $P_m(t)$ is non-decreasing in $m$ for a given $t>0$ and non-increasing in $t$ for a given $m$.

As for cost-benefit analysis, we take $R=13\sqrt{0.6c}$, $C_1=m^{0.2}$, and $C_2=m$. Then the expected net revenue $E_m[NR]$ equals
\begin{align*}
	E_m[NR]&=13\sqrt{0.6c}-m-\frac{m^{0.2}}{E[Y]}.
\end{align*}
In Panel~(b) of Figure~\ref{fig:gamma3D}, the red vertical line indicates that the maximum expected net revenue is $15.95$ and is achieved when $m=22$.  Though $E[T_m]$ and $P_m(t)$ are non-decreasing in $m$,  a  large $m$ is not preferable from the point of view of cost-benefit analysis because an optimal $m$ exists in this case.
\end{example}

\begin{example}[Gamma hacking- and Weibull detecting times]
\label{example3}
In this example, we assume that $Y\sim\mathsf{Weibull}(\alpha, \beta)$. That is, $Y$ has the probability density function
\begin{align*}
	f(y)=\frac{\beta}{\alpha}\left(\frac{y}{\alpha}\right)^{\beta-1}e^{-\left(y/\alpha\right)^\beta},\quad y>0.
\end{align*}
We also assume that $X_1, \ldots, X_n \overset{iid}{\sim} \gam(\eta, \delta)$, and that $Y$ is independent from $X_1,  \ldots, X_n$.  Then $\sum_{i=1}^m X_i \sim \gam(m\eta, \delta)$.  Here we have
\begin{align*}
	E[T_m]=&\frac{\delta^{m\eta}\bigg(\int_0^\infty \int_y^\infty y^\beta e^{-\left(\left(y/\alpha\right)^\beta+\delta s\right)}s^{m\delta -1}dsdy+\int_0^\infty \int_s^\infty y^{\beta - 1}e^{-\left(\left(y/\alpha\right)^\beta+\delta s\right)}s^{m\delta}dyds\bigg)}{\int_0^\infty \gamma(m\eta, \eta s)s^{\beta - 1}e^{-\left(y/\alpha\right)^\beta}ds},\\
	P_m(t)=&e^{-(t/ \alpha)^\beta}-\frac{1}{\Gamma(m\eta)}\gamma(m\eta, \delta t)+\frac{1}{\Gamma(m\eta)}\gamma(m\eta, \delta t)(1-e^{-(t/ \alpha)^\beta})\\
	&+\int_0^t e^{-((t-s)/ \alpha)^\beta}-\frac{1}{\Gamma(m\eta)}\gamma(m\eta, \delta (t-s))+\frac{1}{\Gamma(m\eta)}\gamma(m\eta, \delta (t-s))(1-e^{-((t-s)/ \alpha)^\beta})\\
	&d\sum_{i=1}^\infty \bigg\{\frac{1-e^{-(s/ \alpha)^\beta}}{1- \int_0^\infty \frac{1}{\Gamma(m\eta)}\gamma(m\eta,\delta w)\left(\frac{\beta}{\alpha}\right)\left(\frac{w}{\alpha}\right)^{\beta-1}e^{-\left(w/\alpha\right)^\beta}dw}1_{\{0\leq s \leq\sum_{i=1}^{m}X_i\}}+1_{\{s >\sum_{i=1}^{m}X_i\}}\bigg\}.
\end{align*}
%where $\gamma(\alpha, \lambda)$ is incomplete lower gamma distribution displayed in Example 1. 
For the derivation of the above formulas, see the Appendix.

%%%%%
\begin{figure}%[h]
\begin{center}
\subfigure[$E(T_m)$ as a function of $m$]{\centering \includegraphics[scale=0.45]{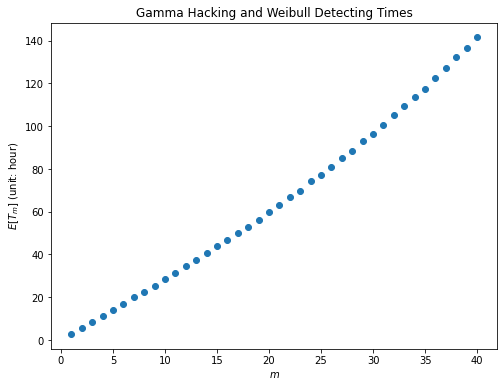}}
\subfigure[$P_m(3)$ as a function of $m$]{\centering \includegraphics[scale=0.45]{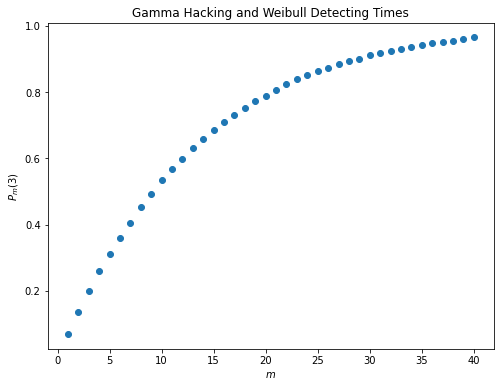}}
\end{center}
\caption{Figure 5: Plots of $E[T_m]$ and $P_m(3)$ under gamma hacking- and weibull detecting times}
\label{fig:weibull}
\end{figure}
%%%%%%
Similar to Example~\ref{example1} and Example~\ref{example2}, we perform a simulation study to estimate $E[T_m]$ and $P_m(t)$.  We use the same method to estimate $E[T_m]$ and $P_m(t)$ as we did in the previous two examples except that here we take $N_1=130, 000$ and $N_2=80, 000$.  Figure~\ref{fig:weibull} plots $E[T_m]$ and $P_m(3)$ as functions of $m$.  

%%%%%%
\begin{figure}%[h]
\begin{center}
	 \subfigure[$P_m(t)$ as a bivariate function of $m$ and $t$]{\centering \includegraphics[width=0.53\textwidth]{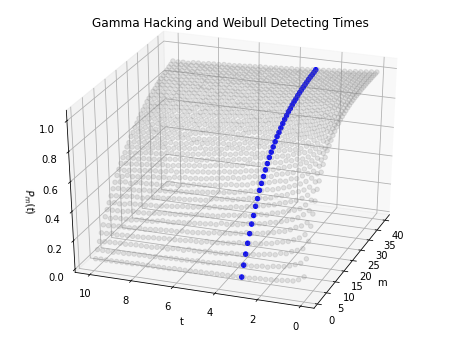}}
	 \subfigure[$E_m(NR)$ as a function of $m$]{\centering \includegraphics[width=0.46\textwidth]{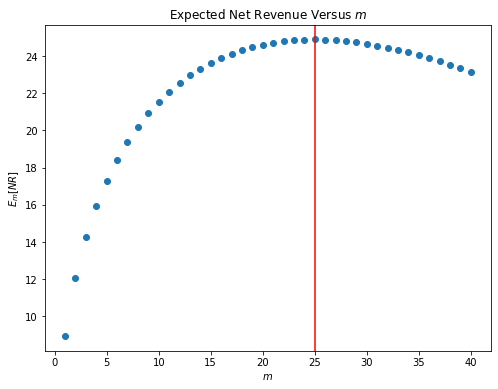}}
\end{center}
\caption{Figure 6: Plots of $P_m(t)$ and  $E_m(NR)$  under Gamma hacking- and Weibull detecting times}
\label{fig:weibull3D}
\end{figure}
%%%%%%
Panel (a) of Figure~\ref{fig:weibull3D} provides a 3D plot of $P_m(t)$ as a bivariate function of $m$ and $t$, where the blue dot scatter plot is the 3D version of the plot in Panel~(b) of Figure~\ref{fig:weibull}. Once again, we see that $P_m(t)$ is non-decreasing in $m$ for a fixed  $t>0$ and non-decreasing in $t$ for a fixed $m$.
 
To perform the cost-benefit analysis, we choose $R=10\sqrt{m}$, $C_1=6m^{0.2}$, and $C_2=m$.  Then $E_m[NR]$ equals
\begin{align*}
	E_m[NR]&=10\sqrt{m}-m-\frac{6m^{0.2}}{E[Y]}.
\end{align*}
Similar to Example~\ref{example2},  here an optimal $m$ exists.  The red vertical line in Panel~(b) of Figure~\ref{fig:weibull3D} indicates this optimal $m$ equals $25$ at which the expected net revenue achieves its maximum value $24.90$.

%It is clear that there is not obvious distinction for expected net revenue from $m=10$ to $m=20$. Thus, if the company would like to increase $P_m(t)$ and, at the same time, have relatively high expected net revenue, $m=20$ would be a good choice. Similarly, if the company doesn't have enough nodes, for example, $m \leq 10$, then $m=10$ can also keep a high expected net revenue.
\end{example}

\section{Concluding remarks}
This paper investigates the instantaneous and limiting behavior of an $n$-node blockchain under cyber attacks from a single hacker.  Under arbitrary hacking and detecting times, we have derived the formulas for the limiting functional probability, instantaneous functional probability, and mean functional time. We have also established that these quantities are all non-decreasing functions of the number of nodes.  These results substantiate the intuition that increasing the number of nodes can effectively improve the security of a blockchain.  To illustrate our results, we also provided several simulated examples to illustrate our results. 

While blockchain is clearly a ``hot topic'' that attracts researchers from various fields,  it seems that the operations-research-theoretic aspect of blockchains has rarely been explored.  This paper aims to stimulate serious interest along this line of research.  For future research, one interesting direction is to investigate the operations-research-theoretic aspect of the blockchain when it faces cyber attacks from several hackers.

%In a more complicated model, one can consider the case where it will take another random amount of time for the blockchain to function anew from the moment it is being re-set.

\section*{Appendix}

For the sake of convenience, we first derive some expressions for the components in (\ref{eq:meantime}) and (\ref{eq:instantaneous2}).  We first look at (\ref{eq:meantime}). 
It is clear that
(\ref{eq:detectprob}) implies that
\begin{equation}
\label{eq:probability}
   p_m=P\bigg\{Y>\sum_{i=1}^m X_i\bigg\}=\int_0^{\infty}F_{\sum_{i=1}^m X_i}(s)f_Y(s)ds.   
\end{equation} 
The type of conditional expectations in (\ref{eq:meantime}) appear frequently in statistical reliability theory and actuarial life contingency theory; see, for example, Bieth et al. (2009, 2010) and Hong (2009).  A rigorous treatment can be found in Hong and Sarkar (2013). 
It follows from Equation~(5) in Hong and Sarkar (2013) that
\begin{eqnarray*}
E[Y_1]&=& E\bigg[Y\bigg|Y<\sum_{i=1}^m X_i\bigg]=\int_0^\infty yf_{Y|Y<\sum_{i=1}^m X_i}(y)dy\\
&=&\frac{\int_0^\infty \int_y^\infty yg(y,s)dsdy}{P(Y<\sum_{i=1}^m X_i)}
=\frac{\int_0^\infty \int_y^\infty yf_Y(y)f_{\sum_{i=1}^m X_i}(s)dsdy}{1-\int_0^{\infty}F_{\sum_{i=1}^m X_i}(s)f_Y(s)ds},
\end{eqnarray*}
where, by independence of $Y$ and $\sum_{i=1}^m X_i$, the joint PDF of $Y$ and $\sum_{i=1}^m X_i$  is given by
\begin{align*}
g(y,s)=f_Y(y)f_{\sum_{i=1}^m X_i}(s),\quad y>0, s>0. 
\end{align*}

Hence, 
\[
	E[N_1][Y_1]=\frac{1-p_m}{p_m}E\bigg[Y\bigg|Y<\sum_{i=1}^m X_i\bigg]\\
	=\frac{\int_0^\infty \int_y^\infty yf_Y(y)f_{\sum_{i=1}^m X_i}(s)dsdy}{\int_0^{\infty}F_{\sum_{i=1}^m X_i}(s)f_Y(s)ds}.
\]
Similarly,
\begin{align*}
	E\bigg[\sum_{i=1}^m X_i\bigg|\sum_{i=1}^m X_i<Y\bigg]=\frac{\int_0^\infty \int_s^\infty sf_Y(y)f_{\sum_{i=1}^m X_i}(s)dyds}{\int_0^{\infty}F_{\sum_{i=1}^m X_i}(s)f_Y(s)ds}.
\end{align*}
It follows that
\begin{equation}
\label{eq:Tm}
	E[T_m]=\frac{\int_0^\infty \int_y^\infty yf_Y(y)f_{\sum_{i=1}^m X_i}(s)dsdy+\int_0^\infty \int_s^\infty sf_Y(y)f_{\sum_{i=1}^m X_i}(s)dyds}{\int_0^{\infty}F_{\sum_{i=1}^m X_i}(s)f_Y(s)ds}.
\end{equation}
Next, we shift our attention to the terms in (\ref{eq:instantaneous2}). We have 
\begin{align*}
P\bigg\{t<Y \wedge \sum_{i=1}^{m}X_i\bigg\}&=1-P\bigg\{\sum_{i=1}^{m}X_i\leq t\bigg\}-P\{Y\leq t\}+P\bigg\{\sum_{i=1}^{m}X_i \leq t,Y\leq t\bigg\}.
\end{align*}
Similarly,
\begin{align*}
	P\bigg\{Y \wedge \sum_{i=1}^{m}X_i>t-s\bigg\}=&1-P\bigg\{\sum_{i=1}^{m}X_i\leq t-s\bigg\}-P\{Y\leq t-s\}\\
	&+P\bigg\{\sum_{i=1}^{m}X_i \leq t-s,Y\leq t-s\bigg\}.
\end{align*}
For the term $G(s)=\sum_{i=1}^{\infty} F_{Y_i}(s)$, we have
\begin{equation}
\label{eq:conditionalprobability}
	F_{Y_i}(s)=\frac{P\bigg\{Y \leq s, Y\leq \sum_{i=1}^{m}X_i\bigg\}}{P\bigg\{Y\leq \sum_{i=1}^{m}X_i\bigg\}}=
	\begin{cases}
		1 & \text{if $s >\sum_{i=1}^{m}X_i$},\\
		\frac{P\{Y \leq s\}}{P\{Y\leq \sum_{i=1}^{m}X_i\}} & \text{if $0\leq s \leq\sum_{i=1}^{m}X_i$},\\
		0 & \text{if $s<0$}.
	\end{cases}
\end{equation}
Thus, (\ref{eq:conditionalprobability}) can also be written as
\begin{align*}
	F_{Y_i}(s)=\frac{P(Y \leq s)}{P(Y\leq \sum_{i=1}^{m}X_i)}1_{\{0\leq s \leq\sum_{i=1}^{m}X_i\}}+1_{\{s >\sum_{i=1}^{m}X_i\}},\quad x>0.
\end{align*}

%\subsection{Exponential hacking- and exponential detecting times}
\subsection*{From Example~\ref{example1}}
Let $Y\sim \ex(\delta)$ and  $\sum_{i=1}^m X_i \sim \gam(m, \lambda)$.  Then
 (\ref{eq:probability}) implies
\begin{align*}
	p_m=P\bigg\{Y>\sum_{i=1}^m X_i\bigg\}&=\int_0^\infty \frac{1}{\Gamma(m)}\gamma(m, \lambda s)\delta e^{-\delta s}ds,
\end{align*}
where $\gamma(m, \lambda s)=\int_0^{\lambda s}t^{m-1}e^{-t}dt$ is the lower incomplete gamma function.
By  (\ref{eq:Tm}), we have
\begin{align*}
	E[T_m]&=\frac{\int_0^\infty\int_y^\infty y\delta e^{-\delta y}\frac{\lambda^m}{\Gamma(m)}s^{m-1}e^{-\lambda s}dsdy+\int_0^\infty \int_s^\infty s\delta e^{-\delta y}\frac{\lambda^m}{\Gamma(m)}s^{m-1}e^{-\lambda s}dyds}{\int_0^\infty \frac{1}{\Gamma(m)}\gamma(m, \lambda s)\delta e^{-\delta s}ds}\\
	&=\frac{\lambda^m \bigg(\int_0^\infty \int_y^\infty ye^{-(\delta y+\lambda s)}s^{m-1}dsdy+\int_0^\infty \int_s^\infty e^{-(\delta y+\lambda s)}s^m dyds\bigg)}{\int_0^\infty \gamma(m,\lambda s)e^{-\delta s}ds}.
\end{align*}
It follows from (\ref{eq:instantaneous2}) that
\begin{align*}
	P_m(t)=&1-\frac{1}{\Gamma(m)}\gamma(m, \lambda t)-(1-e^{-\delta t})+\frac{1}{\Gamma(m)}\gamma(m, \lambda t)(1-e^{-\delta t})\\
	&+\int_0^t \left[1-\frac{1}{\Gamma(m)}\gamma(m, \lambda (t-s))-(1-e^{-\delta (t-s)})+\frac{1}{\Gamma(m)}\gamma(m, \lambda (t-s))(1-e^{-\delta (t-s)}) \right]\\
	&d\sum_{i=1}^\infty \bigg\{\frac{1-e^{-\delta s}}{1-\int_0^\infty \frac{1}{\Gamma(m)}\gamma(m, \lambda w)\delta e^{-\delta w}dw}1_{\{0\leq s \leq\sum_{i=1}^{m}X_i\}}+1_{\{s >\sum_{i=1}^{m}X_i\}}\bigg\}\\
	=&e^{-\delta t}-\frac{1}{\Gamma(m)}\gamma(m, \lambda t)+\frac{1}{\Gamma(m)}\gamma(m, \lambda t)(1-e^{-\delta t})\\
	&+\int_0^t \left[ e^{-\delta (t-s)}-\frac{1}{\Gamma(m)}\gamma(m, \lambda (t-s))+\frac{1}{\Gamma(m)}\gamma(m, \lambda (t-s))(1-e^{-\delta (t-s)}) \right]\\
	&d\sum_{i=1}^\infty \bigg\{\frac{1-e^{-\delta s}}{1-\int_0^\infty \frac{1}{\Gamma(m)}\gamma(m, \lambda w)\delta e^{-\delta w}dw}1_{\{0\leq s \leq\sum_{i=1}^{m}X_i\}}+1_{\{s >\sum_{i=1}^{m}X_i\}}\bigg\}.
\end{align*}

%\subsection{Gamma hacking- and gamma detecting times}
\subsection*{From Example~\ref{example2}}
Here $Y\sim \gam(\alpha, \beta)$ and $\sum_{i=1}^m X_i \sim \gam(m\eta, \delta)$. By (\ref{eq:probability}), we have
\begin{align*}
p_m=P\bigg\{Y>\sum_{i=1}^m X_i\bigg\}=\int_0^\infty \frac{1}{\Gamma(m\eta)}\gamma(m\eta,\delta s)\frac{\beta^\alpha}{\Gamma(\alpha)}s^{\alpha-1}e^{-\beta s}ds,
\end{align*}
where $\gamma(m\eta, \delta s)=\int_0^{\delta s}t^{m\eta-1}e^{-t}dt$ is the lower incomplete gamma function.
It follows from (\ref{eq:Tm}) that
\begin{align*}
E[T_m]&=\frac{\int_0^\infty \int_y^\infty y\frac{\beta^\alpha}{\Gamma(\alpha)}y^{\alpha-1}e^{-\beta y}\frac{\delta^{m\eta}}{\Gamma(m\eta)}s^{m\eta-1}e^{-\delta s}dsdy+\int_0^\infty \int_s^\infty s\frac{\beta^\alpha}{\Gamma(\alpha)}y^{\alpha-1}e^{-\beta y}\frac{\delta^{m\eta}}{\Gamma(m\eta)}s^{m\eta-1}e^{-\delta s}dyds }{\int_0^\infty\frac{1}{\Gamma(m\eta)}\gamma(m\eta,\eta s)\frac{\beta^\alpha}{\Gamma(\alpha)}s^{\alpha-1}e^{-\beta s}ds}\\
&=\frac{\delta^{m\eta}\bigg(\int_0^\infty \int_y^\infty y^{\alpha} e^{-(\beta y + \delta s)}s^{m\eta-1}dsdy+\int_0^\infty \int_s^\infty y^{\alpha-1}e^{-(\beta y+\delta s)}s^{m\eta} dyds\bigg )}{\int_0^\infty \gamma(m\eta, \eta s)s^{\alpha -1}e^{-\beta s}ds}.
\end{align*}
Then (\ref{eq:instantaneous2}) implies
\begin{align*}
	P_m(t)=&1-\frac{1}{\Gamma(m\eta)}\gamma(m\eta, \delta t)-\frac{1}{\Gamma(\alpha)}\gamma(\alpha, \beta t)+\frac{1}{\Gamma(m\eta)}\gamma(m\eta, \delta t)\frac{1}{\Gamma(\alpha)}\gamma(\alpha, \beta t)\\
	&+\int_0^t \bigg[ 1-\frac{1}{\Gamma(m\eta)}\gamma(m\eta, \delta (t-s))-\frac{1}{\Gamma(\alpha)}\gamma(\alpha, \beta (t-s))\\
	&+\frac{1}{\Gamma(m\eta)}\gamma(m\eta, \delta (t-s))\frac{1}{\Gamma(\alpha)}\gamma(\alpha, \beta (t-s)) \bigg]\\
	&d\sum_{i=1}^\infty \bigg\{\frac{\frac{1}{\Gamma(\alpha)}\gamma(\alpha, \beta s)}{1-\int_0^\infty \frac{1}{\Gamma(m\eta)}\gamma(m\eta,\delta w)\frac{\beta^\alpha}{\Gamma(\alpha)}w^{\alpha-1}e^{-\beta w}dw}1_{\{0\leq s \leq\sum_{i=1}^{m}X_i\}}+1_{\{s >\sum_{i=1}^{m}X_i\}}\bigg\}.
\end{align*}

%\subsection{Gamma hacking- and Weibull detecting times}
\subsection*{From Example~\ref{example3}}
Here $Y\sim \mathsf{Weibull}(\alpha, \beta)$ and $\sum_{i=1}^m X_i \sim \gam(m\eta, \delta)$.  It follows from (\ref{eq:probability}) that
\begin{align*}
P\bigg(Y>\sum_{i=1}^m X_i\bigg)=\int_0^\infty \frac{1}{\Gamma(m\eta)}\gamma(m\eta,\delta s)\frac{\beta}{\alpha}(\frac{s}{\alpha})^{\beta-1}e^{-\left(s/\alpha\right)^\beta}ds.
\end{align*}
By (\ref{eq:Tm}), we have
\begin{align*}
E[T_m]=&\frac{\int_0^\infty \int_y^\infty y\frac{\beta}{\alpha}(\frac{y}{\alpha})^{\beta-1}e^{-\left(y/\alpha\right)^\beta}\frac{\delta^{m\eta}}{\Gamma(m\eta)}s^{m\eta-1}e^{-\delta s}dsdy}{\int_0^\infty \frac{1}{\Gamma(m\eta)}\gamma(m\eta,\eta s)\frac{\beta}{\alpha}(\frac{s}{\alpha})^{\beta-1}e^{-\left(s/\alpha\right)^\beta}ds}\\
&+\frac{\int_0^\infty \int_s^\infty s\frac{\beta}{\alpha}(\frac{y}{\alpha})^{\beta-1}e^{-\left(y/\alpha\right)^\beta}\frac{\delta^{m\eta}}{\Gamma(m\eta)}s^{m\eta-1}e^{-\delta s}dyds }{\int_0^\infty \frac{1}{\Gamma(m\eta)}\gamma(m\eta,\eta s)\frac{\beta}{\alpha}(\frac{s}{\alpha})^{\beta-1}e^{-\left(s/\alpha\right)^\beta}ds}\\
=&\frac{\delta^{m\eta}\bigg(\int_0^\infty \int_y^\infty y^\beta e^{-\left(\left(y/\alpha\right)^\beta+\delta s\right)}s^{m\delta -1}dsdy+\int_0^\infty \int_s^\infty y^{\beta - 1}e^{-\left(\left(y/\alpha\right)^\beta+\delta s\right)}s^{m\delta}dyds\bigg)}{\int_0^\infty \gamma(m\eta, \eta s)s^{\beta - 1}e^{-\left(y/\alpha\right)^\beta}ds}.
\end{align*}
Then (\ref{eq:instantaneous2}) implies
\begin{align*}
	P_m(t)=&1-\frac{1}{\Gamma(m\eta)}\gamma(m\eta, \delta t)-(1-e^{-\left(t/ \alpha\right)^\beta})+\frac{1}{\Gamma(m\eta)}\gamma(m\eta, \delta t)(1-e^{-\left(t/ \alpha\right)^\beta})\\
	&+\int_0^t 1-\frac{1}{\Gamma(m\eta)}\gamma(m\eta, \delta (t-s))-(1-e^{-\left(\left(t-s\right)/ \alpha\right)^\beta})\\
	&+\frac{1}{\Gamma(m\eta)}\gamma(m\eta, \delta (t-s))(1-e^{-\left(\left(t-s\right)/ \alpha\right)^\beta})\\
	&d\sum_{i=1}^\infty \bigg\{\frac{1-e^{-\left(s/ \alpha\right)^\beta}}{1- \int_0^\infty \frac{1}{\Gamma(m\eta)}\gamma(m\eta,\delta w)\frac{\beta}{\alpha}(\frac{w}{\alpha})^{\beta-1}e^{-\left(w/\alpha\right)^\beta}dw}1_{\{0\leq s \leq\sum_{i=1}^{m}X_i\}}+1_{\{s >\sum_{i=1}^{m}X_i\}}\bigg\}\\
	=&e^{-\left(t/ \alpha\right)^\beta}-\frac{1}{\Gamma(m\eta)}\gamma(m\eta, \delta t)+\frac{1}{\Gamma(m\eta)}\gamma(m\eta, \delta t)(1-e^{-\left(t/ \alpha\right)^\beta})\\
	&+\int_0^t e^{-\left(\left(t-s\right)/ \alpha\right)^\beta}-\frac{1}{\Gamma(m\eta)}\gamma(m\eta, \delta (t-s))+\frac{1}{\Gamma(m\eta)}\gamma(m\eta, \delta (t-s))(1-e^{-\left(\left(t-s\right)/ \alpha\right)^\beta})\\
	&d\sum_{i=1}^\infty \bigg\{\frac{1-e^{-\left(s/ \alpha\right)^\beta}}{1- \int_0^\infty \frac{1}{\Gamma(m\eta)}\gamma(m\eta,\delta w)\frac{\beta}{\alpha}(\frac{w}{\alpha})^{\beta-1}e^{-\left(w/\alpha\right)^\beta}dw}1_{\{0\leq s \leq\sum_{i=1}^{m}X_i\}}+1_{\{s >\sum_{i=1}^{m}X_i\}}\bigg\}.
\end{align*}

%\section*{Acknowledgments}
%We thank the Associate Editor and two anonymous reviewers for many useful comments and suggestions.

\section*{References}

\begin{description}

\item{} Bieth, B.,~ Hong, L. ~and Sarkar, J. ~(2010). A standby system with two types of repair
persons.  \emph{Applied Stochastic Models in Business and Industry}~26(5), 577--594.

\item{}Bieth, B.,~ Hong, L. ~and Sarkar, J. ~(2010). A standby system with two repair persons
under arbitrary life- and repair times.  \emph{Mathematical and Computer Modeling}~51(5-6),756--767.

\item{} Cao, L.~(2020). AI in Finance: A Review. \url{https://ssrn.com/abstract=3647625}.

\item{} Casino, F., Dasaklis, T.~K.~and Patsakis, C.~(2019). A systematic literature review of blockchain-based applications: current status, classification and open issues. \emph{Telematics and Informatics}~36, 55--81.

\item{} Diffie, W.~and Hellman, M.E.~(1976).  New directions in cryptography. \emph{IEEE Transactions on Information Theory}~22(5), 644--654. 

\item{} Goldwasser, S.~and Micali, S.~(1982). Probabilistic encryption. \emph{Journal of Computer System and Sciences}~28(2), 270--299. 

\item{} Goldwasser, S.~,Micali, S.~and Rivest, R.L.~(1988). A digital signature scheme secure against adaptive chosen-message attacks. \emph{SIAM Journal on Computing}~18(1), 283--308. 

%\item{} Hoel, P.G.~, Port, S.~C.~and Stone, C.~J.~(1971). \emph{Introduction to Probability Theory}. Houghton Mifflin Co.: Boston.

\item{} Hong, L.~(2009).  Limiting performance of a one-unit system under various repair models. \emph{Ph.D.  thesis}.  Purdue University.  \url{https://docs.lib.purdue.edu/dissertations/AAI3379404/}.

\item{} Hong, L.~and Sarkar, J.~(2013). Contingent means in multi-life models. \emph{Scandinavian Actuarial Journal}~\textbf{5}, 340--351.

\item{} Kallenberg, O.~(2002). \emph{Foundations of Modern Probability}, Second Edition. Springer: New York.

\item{} Khan, M.A.~and Salah, K.~(2017). IoT security: review, blockchain solutions, and open challenges.  \emph{Future Generation Computer Systems.}~82, 395--411.

\item{} Kumar, A~, Abhishek, K.~,  Ghalib, M.R.,  Nerurkar, P., Bhirud, S.~, Alnumay, W., Kumar, S.A., Chatterjee, P.~and Ghosh, U.~(2020).  Securing logistics system and supply chain using Blockchain.\emph{Applied stochastic models in business and industry}~37(3), 413--428. 

\item{} Li, C.~and Zhang, L.J.~(2017). A blockchain based new secure multi-layer network model for internet of things.  \emph{Proceedings--2017 IEEE 2nd Intertional Congress on Internet of Things}, 33--41. 

\item{} Lu, J.~. Wu, S.~, Cheng, H.~, Song, B.~and Xiang, Z.~(2020). Smart contract for electricity transactions and charge settlements using blockchain.\emph{Applied stochastic models in business and industry}~37(3), 442--453.

\item{} Malik, N.,~ Alkhatib, K.,~ Sun, Y.,~ Knight, E.~and Jararweh, Y.~(2021). A comprehensive review of blockchain applications in industrial Internet of Things. \emph{Applied stochastic models in business and industry}~37(3), 391--412.

\item{} Meng, W., Tischhauser, E.~, Wang, Q.~, Wang, Y.~and Han, J.~(2017). When instruction detection meets blockchain technology: a review. \emph{IEEE Access}~6, 10179--10188.

\item{} Risius, M.,~and Spohrer, K.~(2017).  A Blockchain Research Framework: What We (don't) Know, Where We Go from Here, and  How We Will Get There.  \emph{Business $\&$ Information Systems Engineering volume}~59(6), 385--409.

\item{} Rivest, R.L.~, Shamir, A.~and Adelman, L.~(1978). A method for obtaining digital signatures and public-key cryptosystems.  \emph{Communications of the ACM}~21(2), 120--126.

%\item{} Ross, S.~M.~(1996). \emph{Stochastic Processes}, Second Edition. Wiley: New York.

\item{} Salah, K., Rehman, M., Nizamuddin, N.,~and Al-fuqaha, A. Blockchain for AI: Review and Open Research Challenges. \emph{IEEE Access}~10127-10149.

\item{} Tama, B.A.~, Kweka, B.J.,~ Park, Y.~and Rhee, K.H.~(2017).  A critical review of blockchain and its current applications.  \emph{2017 International Conference on Electrical Engineering and Computer Science IEEE}, 109--113. 

\item{} Yang, J.~, Ma, X.~, Crespo,   R. G.~and  Martínez, O.S.~(2020). Blockchain for supply chain performance and logistics management.\emph{Applied stochastic models in business and industry}~37(3), 429--441.

\item{} Zheng, Z.~Xie, S.~, Dai, H.N., and Wang, H.~(2018). Blockchain challenges and opportunities: a survey.  \emph{International Journal of Web and Grid Services}~14(4), 352--375.

\end{description}
\end{document}